\newif\iffull
\fulltrue
\iffull
\documentclass[a4paper,UKenglish,cleveref, autoref,
thm-restate,authorcolumns]{lipics-v2019}
\nolinenumbers

\title{Asymptotic Approximation by Regular Languages}

\author{Ryoma Sin'ya}{Akita University, Akita, Japan \and RIKEN AIP, Japan}{ryoma@math.akita-u.ac.jp}{https://orcid.org/0000-0002-8152-998X}{JSPS KAKENHI Grant Number JP19K14582}

\authorrunning{R.\, Sin'ya} 

\Copyright{Ryoma Sin'ya} 



\category{} 

\relatedversion{} 

\supplement{}

\EventEditors{John Q. Open and Joan R. Access}
\EventNoEds{2}
\EventLongTitle{42nd Conference on Very Important Topics (CVIT 2016)}
\EventShortTitle{CVIT 2016}
\EventAcronym{CVIT}
\EventYear{2016}
\EventDate{December 24--27, 2016}
\EventLocation{Little Whinging, United Kingdom}
\EventLogo{}
\SeriesVolume{42}
\ArticleNo{23}

\usepackage{url}
\usepackage{ascmac,bm}
\usepackage{amsmath,amssymb}
\usepackage{extarrows}

\theoremstyle{plain}

\newtheorem{problem}[theorem]{Problem}
\else
\documentclass{llncs}

\title{Asymptotic Approximation\\ by Regular Languages}
\author{Ryoma Sin'ya\inst{1}}
\institute{Akita University\\
\email{ryoma@math.akita-u.ac.jp}\footnote{The author is also with RIKEN
AIP.}}

\usepackage[dvipdfmx]{graphicx}
\usepackage{url}
\usepackage{ascmac,bm}
\usepackage{amsmath,amssymb}

\usepackage{float}
\usepackage{extarrows}
\usepackage{xcolor}

\fi

\newcommand{\CA}{{\cal A}}
\newcommand{\CJ}{{\cal J}}
\newcommand{\CL}{{\cal L}}
\newcommand{\CR}{{\cal R}}
\newcommand{\CH}{{\cal H}}

\newcommand{\pr}{\mathrm{Pr}}

\newcommand{\ie}{{\it i.e.}}
\newcommand{\eg}{{\it e.g.}}

\newcommand{\cf}{{\it cf. }}

\newcommand{\defeq}{\xlongequal{\!\!\!\mathtt{def}\!\!\!}}

\newcommand{\nat}{\mathbb{N}}

\newcommand{\diff}{\triangle}

\newcommand{\prim}{\mathsf{Q}}
\newcommand{\gold}{\mathsf{G}}
\newcommand{\dyck}{\mathsf{D}}
\newcommand{\pal}{\mathsf{P}}

\newcommand{\kemp}{\mathsf{K}}
\newcommand{\othree}{\mathsf{O}_3}
\newcommand{\ofour}{\mathsf{O}_4}

\newcommand{\maj}{\mathsf{M}}

\newcommand{\rev}{\mathrm{rev}}

\renewcommand{\alph}{\mathtt{Alph}}

\newcommand{\card}[1]{\#\!\left(#1\right)}

\newcommand{\CC}{{\cal C}}
\newcommand{\CD}{{\cal D}}

\newcommand{\reg}{\mathrm{REG}}
\newcommand{\cfl}{\mathrm{CFL}}
\newcommand{\dcfl}{\mathrm{DetCFL}}
\newcommand{\ucfl}{\mathrm{UnCFL}}

\newcommand{\slo}{<_{\text{lex}}}

\newcommand{\pd}[1]{\delta_{#1}}
\renewcommand{\d}{\pd{A}}
\newcommand{\pcd}[1]{\delta^*_{#1}}
\newcommand{\cd}{\pcd{A}}
\renewcommand{\pm}[1]{\mu_{#1}}
\newcommand{\plm}[1]{\underline{\mu}_{#1}}
\newcommand{\pum}[1]{\overline{\mu}_{#1}}
\newcommand{\lm}{\plm{\reg}}
\newcommand{\um}{\pum{\reg}}
\newcommand{\m}{\pm{\reg}}

\newcommand{\qedd}{\iffull
\else
\qed
\fi}

\begin{document}

\maketitle

\begin{abstract}
This paper investigates a new property of formal languages called
 \(\reg\)-measurability where \(\reg\) is the class of regular languages.
Intuitively, a language \(L\) is \(\reg\)-measurable if there exists an
 infinite sequence of regular languages that ``converges'' to \(L\).
A language without \(\reg\)-measurability has a complex shape in some
 sense so that it can not be (asymptotically) approximated by regular
 languages.
We show that several context-free languages are \(\reg\)-measurable (including languages with transcendental generating function and
 transcendental density, in particular), while
 a certain simple deterministic context-free language and the set of
 primitive words are \(\reg\)-immeasurable in a strong sense.
\end{abstract} 
\section{Introduction}
Approximating a complex object by more simple objects is a major concept
in both computer science and mathematics.
In the theory of formal languages, various types of approximations have
been investigated (\eg,
\cite{kappe2004,kappe2005,approx,CORDY2007125,CAMPEANU20013,Domaratzki}).
For example, Kappes and Kintala~\cite{kappe2004} introduced 
\emph{convergent-reliability} and \emph{slender-reliability} which
measure how a given deterministic automaton \(\CA\) nicely approximates a given
language \(L\) over an alphabet \(A\).
Formally \(\CA\) is said to accept \(L\) convergent-reliability if
 \iffull
 the
 ratio of the number of \emph{in}correctly accepted/rejected words of
 length \(n\)
\[
 \card{(L(\CA) \diff L) \cap A^n} /
 \card{A^n}
\]
tends to 0 if \(n\) tends to infinity,
\fi
and is said to accept \(L\) slender-reliability if
the number of incorrectly accepted/rejected words of length \(n\) is always bounded
above by some constant \(c\): \ie, \(\card{(L(\CA) \diff L) \cap A^n}
\leq c\) for any \(n\).
Here \(L(\CA)\) denotes the language accepted by \(\CA\), \(\card{S}\)
denotes the cardinality of the set \(S\), \(\overline{L}\) denotes the
complement of \(L\) and \(\diff\) denotes the symmetric difference.
A slightly modified version of approximation is
\emph{bounded-\(\epsilon\)-approximation} which was introduced by Eisman
and Ravikumar.
They say that two languages \(L_1\) and \(L_2\) provide a
bounded-\(\epsilon\)-approximation of language \(L\) if
\(L_1 \subseteq L \subseteq L_2\) holds and
 the ratio of their length-\(n\) difference satisfies
\iffull
\[
 \card{(L_2 \setminus L_1) \cap A^n} / \card{A^n} \leq \epsilon
\]
\fi
for every sufficiently large \(n \in \nat\).
Perhaps surprisingly, they showed that no pair of regular languages can
 provide a bounded-\(\epsilon\)-approximation of the
language \(\{w \in \{a,b\}^* \mid w \text{ has more } a \text{'s than } b
\text{'s}\}\) for any \(0 \leq \epsilon < 1\)~\cite{approx}.
This result is a very strong \emph{in}approximable (by regular languages)
example of certain non-regular languages.
Also, there is a different framework of approximation so-called
\emph{minimal-cover}~\cite{Domaratzki,CAMPEANU20013}, and
a notion represents some \emph{in}approximability by regular languages
so-called \emph{\(\reg\)-immunity}~\cite{immunity}.

A model of approximation introduced in this paper
  is rather close to the work of Eisman and Ravikumar~\cite{approx}.
Instead of approximating by a \emph{single} regular language, we consider an
  approximation of some non-regular language \(L\) by an \emph{infinite
  sequence}
  of regular languages that ``converges'' to \(L\).
Intuitively, we say that \(L\) is \emph{\(\reg\)-measurable} if there exists
  an infinite sequence of pairs of regular languages \((K_n, M_n)_{n
  \in \nat}\) such that
  \(K_n \subseteq L \subseteq M_n\) holds for all \(n\) and 
  the ``size'' of the difference \(M_n \setminus K_n\) tends to \(0\) if
  \(n\) tends to infinity.
  The formal definition of ``size'' is formally described in the next
  section: we use a notion called \emph{density (of languages)} for measuring the ``size'' of a language.

Although we used the term ``approximation'' in the title and there are
various research on this topic in formal language theory,
our work is strongly influenced by the work of Buck~\cite{Buck}
which investigates, as the title said, \emph{the measure theoretic approach to density}.
In~\cite{Buck} the concept of \emph{measure density} \(\mu\) of
subsets of natural numbers \(\nat\) was introduced.
Roughly speaking, Buck considered an arithmetic progression \(X = \{cn + d \mid n \in \nat \}\)
(where \(c, d \in \nat\), \(c\) can be zero) as a ``basic set'' whose \emph{natural density} as
\(\delta(X) = 1/c\) if \(c \neq 0\) and \(\delta(X) = 0\) otherwise, then defined the \emph{outer measure density}
\(\mu^*(S)\) of any subset \(S \subseteq \nat\) as
\begin{align*}
\mu^*(S) = \inf \Bigl\{ \sum_{i} \delta(X_i) \mid S \subseteq X
 & \text{ and } X \text{ is a finite union of }\\
 & \text{ disjoint arithmetic progressions } X_1, \ldots, X_k \Bigr\}.
\end{align*}
Then the \emph{measure density} \(\mu(S) = \mu^*(S)\) was introduced
for the sets satisfying the condition
\iffull
\begin{align}
 \mu^*(S) + \mu^*(\overline{S}) = 1
\end{align}
\fi
where \(\overline{S} = \nat \setminus S\).
Technically speaking, the class \(\CD_\mu\) of all subsets of natural numbers satisfying
Condition~(1) is the \emph{Carath\'eodory extension} of the class
\[
 \CD_0 \defeq \{ X \subseteq \nat \mid X \text{ is a finite
 union of arithmetic progressions } \},
\]
see Section~2 of \cite{Buck} for more details.
Notice that here we regard a singleton \(\{d\}\) as an arithmetic progression
(the case \(c = 0\) for \(\{ cn + d \mid n \in \nat \}\)), any finite set
belongs to \(\CD_0\).
Buck investigated several properties of \(\mu\) and \(\CD_\mu\), and
showed that \(\CD_\mu\) \emph{properly} contains \(\CD_0\).

In the setting of formal languages, it is very natural to consider
the class \(\reg\) of regular languages as ``basic sets'' since it has
various types of representation, good closure properties and rich decidable properties.
Moreover, if we consider regular languages \(\reg_A\) over a unary alphabet \(A =
\{a\}\), then \(\reg_A\) is isomorphic to the class \(\CD_0\);
it is well known that the Parikh image \(\{|w| \mid w \in L\} \subseteq
\nat\) (where \(|w|\) denotes the
 length of \(w\)) of every regular language
 \(L\) in \(\reg_A\) is semilinear and hence it
 is just a finite union of arithmetic progressions.
From this observation, investigating the densities of regular languages and
its measure densities (\ie, \(\reg\)-measurability) for non-regular languages can be naturally
considered as an adaptation of Buck's study~\cite{Buck} for formal
language theory.

\subsection*{Our contribution}
In this paper we investigate \(\reg\)-measurability
(\(\simeq\) asymptotic approximability by regular languages) of
non-regular, mainly context-free languages.
The main results consist of three kinds. We show that: (1) several
context-free languages (including languages with \emph{transcendental generating function} and
 \emph{transcendental density}) are \(\reg\)-measurable [Theorem~\ref{thm:ab}--\ref{thm:suff}].
 (2) there are ``very large/very small'' (deterministic) context-free languages that
 are \(\reg\)-immeasurable in a strong sense [Theorem~\ref{thm:pmaj}].
 (3) the set of \emph{primitive words} is ``very large'' and
 \(\reg\)-immeasurable in a strong sense
 [Theorem~\ref{thm:prim1}--\ref{thm:prim2}].
Open problems and some possibility of an application of the notion of measurability to
classifying formal languages will be stated in Section~\ref{sec:fut}.

\iffull
The paper is organised as follows.
Section~\ref{sec:pre} provides mathematical background of densities of
formal languages.
The formal definition of \(\reg\)-approximability and
\(\reg\)-measurability are introduced in Section~\ref{sec:measure}.
The scenario of Section~\ref{sec:measure} mostly follows one of the
measure density introduced by Buck~\cite{Buck} which was described above.
In Section~\ref{sec:cfl}, we will give several examples of \(\reg\)-inapproximable but
\(\reg\)-measurable context-free languages.
These examples include, perhaps somewhat surprisingly, a language with a
\emph{transcendental density} which have been considered as a very complex
context-free language from a combinatorial viewpoint.
In Section~\ref{sec:prim}, we consider the set of so-called
\emph{primitive words} and its \(\reg\)-measurability.
Section~\ref{sec:fut} ends this paper with concluding remarks, some
future work and open problems.
\fi
We assume that the reader has a basic knowledge of formal language
theory. \section{Densities of Formal Languages}\label{sec:pre}
For a set \(S\), we write \(\card{S}\) for the cardinality of \(S\). The set of
natural numbers including \(0\) is denoted by \(\nat\).
For an alphabet $A$, we denote the set of all words (resp. all non-empty
words) over $A$ by $A^*$ (resp. $A^+$).
We write \(\varepsilon\) for the empty word and write $A^n$
(resp. $A^{<n}$) for the set of all words of length $n$
(resp. less than $n$).
For a language \(L\), we write \(\alph(L)\) for the set of all letters
appeared in \(L\).
For word $w \in A^*$ and a letter \(a \in A\), $|w|_a$ denotes the number of
occurrences of $a$ in $w$. A word \(v\) is said to be a \emph{factor} of
a word \(w\) if \(w = xvy\) for some \(x,y \in A^*\), further said to be a
\emph{prefix} of \(w\) if \(x = \varepsilon\).
For a language \(L \subseteq A^*\), we denote by \(\overline{L} = A^*
\setminus L\) the complement of \(L\).

A \emph{language class} \(\CC\) is a family of languages \(\{\CC_A\}_{A: \text{
finite alphabet}}\) where \(\CC_A \subseteq 2^{A^*}\) for each \(A\)
and \(\CC_A \subseteq \CC_B\) for each \(A \subseteq B\).
We simply write \(L \in \CC\) if \(L \in \CC_A\) for some alphabet
\(A\).
We denote by \(\reg, \dcfl, \ucfl\) and \(\cfl\) the class of regular languages,
deterministic context-free languages, unambiguous context-free languages
and context-free languages, respectively.
A language \(L\) is said to be \emph{\(\CC\)-immune} if \(L\) is
infinite and no infinite subset of \(L\) belongs to \(\CC\).

\begin{definition}\upshape
Let \(L \subseteq A^*\) be a language.
The \emph{natural density} \(\d(L)\) of \(L\) is defined as
\[
 \d(L) \defeq \lim_{n \rightarrow \infty} \frac{\card{L \cap A^n}}{\card{A^n}}
\]
 if the limit exists, otherwise we write \(\d(L) = \bot\) and say that
 \(L\) does not have a natural density.
The \emph{density} \(\cd(L)\) of \(L\) is defined as
\[
 \cd(L) \defeq \lim_{n \rightarrow \infty} \frac{1}{n} \sum_{k = 0}^{n-1} \frac{\card{L \cap A^k}}{\card{A^k}}
\]
 if its exists, otherwise we write \(\cd(L) = \bot\) and say that
 \(L\) does not have a density.
 A language \(L \subseteq A^*\) is called \emph{null} if \(\cd(L) = 0\),
 and conversely \(L\) is called \emph{co-null} if \(\cd(L) = 1\).
\end{definition}
\iffull
\begin{remark}
\fi
Notice that
if \(L\) has a natural density (\ie, \(\d(L) \neq \bot\)), then
it also has a density and \(\cd(L) = \d(L)\) holds.
But the converse is not true in general, \eg, the case \(L = (AA)^*\)
 (see Example~\ref{ex:density} below).
\iffull
 \end{remark}
\fi
The following observation is basic.
\begin{claim}\label{claim:basic}
 Let \(K, L \subseteq A^*\) with \(\cd(K) = \alpha, \cd(L) =
 \beta\). Then we have:
\iffull
\begin{enumerate}
 \item \(\alpha \leq \beta\) if \(K \subseteq L\).
 \item \(\cd(L \setminus K) = \beta - \alpha\) if \(K \subseteq L\).
 \item \(\cd(\overline{K}) = 1 - \alpha\).
 \item \(\cd(K \cup L) \leq \alpha + \beta\) if \(\cd(K \cup L) \neq \bot\).
 \item \(\cd(K \cup L) = \alpha + \beta\) if \(K \cap L = \emptyset\).
\end{enumerate}
\fi
\end{claim}
For more properties of \(\cd\), see Chapter~13 of \cite{codes}.

\begin{example}\label{ex:density}
Here we enumerate a few examples of densities of languages.
\begin{itemize}
\iffull
 \item The set of all words $A^*$ clearly satisfies
	   $\d(A^*) = 1$, and its complement $\emptyset$ satisfies
       $\d(\emptyset) = 0$. It is also clear that every finite language
       is null.
  \fi
 \item For the set $\{a\}A^*$ of all words starting with \(a \in A\), we have
       \(
       \card{\{a\}A^* \cap A^n}/\card{A^n} = \card{aA^{n-1}}/\card{A^n} = 1/\card{A}.
       \)
	   Hence $\d(\{a\}A^*) = 1/\card{A}$.
 \item
\iffull
	  Consider $(AA)^*$ the set of all words with even length. Because
	   \[
       \frac{\card{(AA)^* \cap A^n}}{\card{A^n}} = \begin{cases}
						 1 & \text{if} \;\; n \;\; \text{is even,}\\
						 0 & \text{if} \;\; n \;\; \text{is odd.}
					      \end{cases}
	   \]
       holds, its limit does not exist and thus \((AA)^*\) does not have a natural density $\d((AA)^*) = \bot$.
	  However, it has a density \(\cd((AA)^*) = 1/2\).
	  \fi
 \item The semi-Dyck language
       \[
       \dyck \defeq 
       \{ w \in \{a,b\}^* \mid |w|_a = |w|_b \text{ and } |u|_a \geq |u|_b \text{ for every
       prefix } u \text{ of } w \}
       \]
       is non-regular but context-free.
       It is well known that the number of words in \(\dyck\) of length
       \(2n\) is equal to the \(n\)-th Catalan number whose
       asymptotic approximation is \(\Theta(4^n/n^{3/2})\).
	   Thus
\iffull
       \[
       \frac{\card{\dyck \cap A^n}}{\card{A^n}} = \begin{cases}
						   \Theta(1/(n/2)^{3/2}) & \text{if} \;\; n \;\; \text{is even,}\\
						   0 & \text{if} \;\; n \;\; \text{is odd.}
						   \end{cases}
       \]
       and
\fi
	   we have \(\d(\dyck) = 0\), \ie, \(\dyck\) is null.
 \end{itemize}
\end{example}
Example~\ref{ex:density} shows us that, for some regular language $L$,
its natural density is either zero or one, for some, like $L =
\{a\}A^*$ (for \(\card{A} \geq 2\)),
 $\d(L)$ could be a real number strictly between zero and one, and for
 some,
 like $L = (AA)^*$, a natural density may not even exist. 
However, the following theorem tells us that all regular languages
\emph{do} have densities.
 
\begin{theorem}[\cf Theorem~III.6.1 of \cite{Salomaa:1978:ATA:578607}]\label{thm:salomaa}
Let \(L \subseteq A^*\)  be a regular language.
Then there is a positive integer \(c\) such that for
 all natural numbers \(d < c\), 
 \iffull
 the following limit exists
 \[
  \lim_{n \rightarrow \infty} \frac{\card{L \cap A^{cn + d}}}{\card{A^{cn + d}}}
 \]
 \fi
 and it is always rational,
 \ie,
 the sequence \((\card{L \cap A^n}/\card{A^n})_{n \in \nat}\) has only
 finitely many accumulation points and these are rational and periodic.
\end{theorem}
\begin{corollary}\label{cor:ratden}
Every regular language has a density and it is rational.
\end{corollary}
\begin{corollary}\label{cor:null}
For any regular language \(L \subseteq A^*\),
 \(\d(L) = 0\) if and only if \(\cd(L) = 0\).
\end{corollary}
Furthermore, for \emph{unambiguous} context-free languages, the
following holds.
 \begin{theorem}[Berstel~\cite{berstel:hal-00619884}]\label{thm:ber}
  For any unambiguous context-free language \(L\) over \(A\), its
  density \(\cd(L)\), if it exists (\ie, \(\cd(L) \neq \bot\)),
 is always algebraic.
\end{theorem}
In the next section we will introduce a language with a transcendental
density, which should be inherently ambiguous due to Theorem~\ref{thm:ber}.

We conclude the section by introducing the notion called
\emph{dense}: a property about some topological ``largeness'' of a
language (\cf Chapter~2.5 of \cite{codes}).
\begin{definition}\label{def:null}\upshape
A language \(L \subseteq A^*\) is said to be \emph{dense} if the set of
 all factors of  \(L\) is equal to \(A^*\). 
We say that a word \(w \in A^*\) is a \emph{forbidden word}
 (resp. \emph{forbidden prefix}) of \(L\)
 if \(L \cap A^* w A^* = \emptyset\)
 (resp. \(L \cap w A^* = \emptyset\)).
\end{definition}
Observe that \(L \subseteq A^*\) is dense if and only if no word is a forbidden word of \(L\).
The next theorem ties two different notions of ``largeness'' of
languages in the regular case.
\begin{theorem}[S.~\cite{Ryoma}]\label{thm:null}
A regular language is non-null if and only if it is dense.
\end{theorem}
The ``only if''-part of Theorem~\ref{thm:null} is nothing but the well-known so-called
\emph{infinite monkey theorem} (which states that \(L\) is not dense
implies \(L\) is null), and this part is true for any (non-regular) languages.
But we stress that ``if''-part is \emph{not true} beyond regular languages;
for example the semi-Dyck language \(\dyck\) is null \emph{but dense}
(which will be described in Proposition~\ref{prop:dyck}).
We denote by \(\reg^+\) the family of
 non-null regular languages, which is equivalent to the family of regular
 languages with positive densities
 thanks to Corollary~\ref{cor:ratden}. \section{Approximability and Measurability}\label{sec:measure}
Although we will mainly consider \(\reg\)-measurability of non-regular
languages in this paper, here we define two notions approximability and
measurability in general setting, with few concrete examples.

\begin{definition}\upshape
Let \(\CC, \CD\) be classes of languages.
A language \(L\) is said to be
\emph{\((\CC, \epsilon)\)-lower-approximable} 
 if there exists \(K \in \CC\)
 such that \(K \subseteq L\) and
 \(\pcd{\alph(L)}(L \setminus K) \leq \epsilon\).
A language \(L\) is said to be
\emph{\((\CC, \epsilon)\)-upper-approximable} 
 if there exists \(M \in \CC\)
 such that \(L \subseteq M\) and
 \(\pcd{\alph(M)}(M \setminus L) \leq \epsilon\).
 A language \(L\) is said to be
 \emph{\(\CC\)-approximable} 
 if \(L\) is both \((\CC,0)\)-lower and \((\CC,0)\)-upper-approximable.
 \(\CD\) is said to be \(\CC\)-approximable if every language
 in \(\CD\) is \(\CC\)-approximable.
\end{definition}

The following proposition gives a
simple \(\reg\)-inaproximable example.
\begin{proposition}\label{prop:dyck}
The semi-Dyck language \(\dyck\) is \(\reg\)-inapproximable.
\end{proposition}
\begin{proof}
We already mentioned that \(\dyck\) is null in Example~\ref{ex:density}, and thus \(\dyck\) is \((\reg, 0)\)-lower-approx by
 \(\emptyset \subseteq \dyck\).
One can easily observe that \(\dyck\) has no forbidden word: since for any \(w \in A^*\)
 there exists a pair of natural numbers \((n,m) \in \nat^2\) such that
 \(a^n w b^m \in \dyck\).
Hence if a regular language \(L\) satisfies \(\dyck \subseteq L\), \(L\)
 has no forbidden word, too, and thus
 \(L\) is non-null by
 Theorem~\ref{thm:null}.
 Thus by Claim~\ref{claim:basic},
 \(\cd(L \setminus \dyck) = \cd(L) - \cd(\dyck) = \cd(L) > 0\),
 which means that \(\dyck\) can not be \((\reg,0)\)-upper-approximable. \qedd
\end{proof}
The proof of Proposition~\ref{prop:dyck} only depends on the
non-existence of forbidden words, hence we can apply the same proof to the next theorem.
\begin{theorem}\label{thm:nullapprox}
Any null language having no forbidden word is \((\reg,
 0)\)-upper-inapproximable.
\end{theorem}
Because \(\dyck\) is deterministic context-free, in our term we have:
\begin{corollary}
\(\dcfl\) is \(\reg\)-inapproximable.
\end{corollary}

Furthermore, by the combination of Theorem~\ref{thm:ber} and the next theorem, we
will know that there exists a context-free language which can not be
approximated by any unambiguous context-free language.
\begin{theorem}[Kemp~\cite{Kemp1980}]\label{thm:kemp}
Let \(A = \{a,b,c\}\).
Define
\[
 S_1 \defeq \{a\} \{b^i a^i \mid i \geq 1\}^*
 \qquad \qquad
 S_2 \defeq \{a^i b^{2i} \mid i \geq 1\}^* \{a\}^+,
\]
 and
 \[
 L_1 \defeq S_1 \{c\} A^*
 \qquad \qquad\qquad\qquad\!\!
 L_2 \defeq S_2 \{c\} A^*.
 \]
Then \(\kemp \defeq L_1 \cup L_2\) is a context-free language with a transcendental natural density
\(\d(\kemp)\).
\end{theorem}

\begin{corollary}
\(\cfl\) is \(\ucfl\)-inapproximable.
\end{corollary}

We then introduce the notion of \(\CC\)-measurability which is a formal
language theoretic analogue of Buck's measure density~\cite{Buck}.
\begin{definition}\upshape
Let \(\CC, \CD\) be classes of languages.
For a language \(L\), we define its
 \emph{\(\CC\)-lower-density} as
\[
 \plm{\CC}(L) \defeq \sup \{ \cd(K) \mid A = \alph(L), K \subseteq L, K
 \in \CC_A, \cd(K) \neq \bot \}
\]
 and its \emph{\(\CC\)-upper-density}
 as
\[
 \pum{\CC}(L) \defeq \inf \{ \cd(K) \mid A = \alph(L), L \subseteq K, K
 \in \CC_A, \cd(K) \neq \bot \}.
\]
A language \(L\) is said to be
 \emph{\(\CC\)-measurable} if
 \(\pum{\CC}(L) = \plm{\CC}(L)\) holds, and we simply write
 \(\pum{\CC}(L)\) as \(\pm{\CC}(L)\).
\(\CD\) is said to be \(\CC\)-measurable if every language
 in \(\CD\) is \(\CC\)-measurable.
\end{definition}
\begin{definition}\upshape
We call \(\pum{\CC}(L) - \plm{\CC}(L)\) the \emph{\(\CC\)-gap} of a
 language \(L\).
We say that a language \emph{\(L\) has full \(\CC\)-gap}
if its \(\CC\)-gap equals to \(1\), \ie, \(\pum{\CC}(L) - \plm{\CC}(L) = 1\).
\end{definition}

In the next section, we describe several examples of both
\(\reg\)-measurable and \(\reg\)-immeasurable languages.
The \(\reg\)-gap could be a good measure how much a given language has
 a complex shape from the viewpoint of regular languages.

The following lemmata are basic.
\begin{lemma}\label{lem:measure}
Let \(K, L\) be two languages.
\begin{enumerate}
\item \(\pum{\CC}(K) \leq \pum{\CC}(L)\) if \(K \subseteq L\).
\item \(\pum{\CC}(K \cup L) \leq \pum{\CC}(K) + \pum{\CC}(L)\) if \(\CC\) is closed under
      union.
 \item \(\pum{\CC}(K) = \cd(K)\) if \(K \in \CC\) and \(\cd(K) \neq \bot\).
\end{enumerate}
\end{lemma}

\begin{lemma}\label{lem:carathe}
Let \(\CC\) be a language class
 such that \(\CC\) is closed under complement and every language in \(\CC\) has a density.
A language \(L \subseteq A^*\) is \(\CC\)-measurable if and only if
\begin{align}
 \pum{\CC}(L) + \pum{\CC}(\overline{L}) = 1.\label{con:comp}
\end{align}
\end{lemma}
\begin{proof}
Let \(L\) be a language and \(A = \alph(L)\).
By definition, \(L\) satisfies Condition~\eqref{con:comp} if and only if
\begin{align}
 \inf \{ \cd(K) \mid L \subseteq K, K \in \CC \}
 = 1 - \inf \{ \cd(K) \mid \overline{L} \subseteq K, K \in \CC \} \label{con:comp2}
\end{align}
holds. On the other hand, \(L\) is measurable if and only if
\begin{align}
 \inf \{ \cd(K) \mid L \subseteq K, K \in \CC \}
 = \sup \{ \cd(K) \mid K \subseteq L, K \in \CC \}. \label{con:comp3}
\end{align}
For any language \(K \in \CC_A\) such that
 \(K \subseteq L\) and \(\cd(K) \neq \bot\), its complement \(\overline{K}\) satisfies
 \(\overline{L} \subseteq \overline{K}\) and \(\cd(\overline{K}) = 1 -
 \cd(K)\).
This means that if \(\CC_A\) is closed under complement then
 \(
 \sup \{ \cd(K) \mid K \subseteq L, K \in \CC_A \} = 1 - \inf \{ \cd(K) \mid \overline{L} \subseteq K, K \in \CC_A \},
 \) holds, which immediately implies the equivalence of Condition~\eqref{con:comp2} and Condition~\eqref{con:comp3}.
\end{proof}

 \section{\(\reg\)-measurability on Context-free Languages}\label{sec:cfl}
In this section we examine \(\reg\)-measurability of several types of context-free languages.
The first type of languages (Section~\ref{sec:null}) is null context-free languages.
Although some null language can have a full \(\reg\)-gap as stated in
the next 
theorem,
we will show
that typical null context-free languages are \(\reg\)-measurable.
\begin{theorem}\label{thm:hardnull}
There is a recursive language \(L\) which is null but
\(\um(L) = 1\).
\end{theorem}

\begin{proof}
Let \(A\) be an alphabet with \(\card{A} \geq 2\) and let \((\CA_i)_{i \in \nat}\) be an
 enumeration of automata over \(A\)
 such that \(\reg_A = \{ L(\CA_i) \mid i \in \nat\}\); we can take such
 enumeration by enumerating some binary representation of automata via
 shortlex order \(\slo\).
We will construct a null language \(L\) such that \(\um(L) = 1\), in
 particular, \(L\) is not a subset of every regular co-infinite language.

Consider the following program \(P\) which takes an input word \(w\):
\begin{description}
 \item[Step~1] set \(i = 0\) and \(\ell = 0\).
 \item[Step~2] check \(L(\CA_i)\) is co-infinite (\ie,
			the complement \(\overline{L(\CA_i)}\) is infinite) or not.
 \item[Step~3] if \(L(\CA_i)\) is co-finite, then set \(i = i + 1\) and go back to Step~2.
 \item[Step~4] otherwise, pick \(u\) such that \(u\) is the smallest
	    (with respect to \(\slo\)) word satisfying
	    \(|u| > \ell\) and \(u \notin L(\CA_i)\) (such \(u\) surely exists
	    since \(L(\CA_i)\) is co-infinite).
 \item[Step~5] if \(w = u\) then \(P\) accepts \(w\) and halts.
 \item[Step~6] if \(w \slo u\) then \(P\) rejects \(w\) and halts.
 \item[Step~7] if \(u \slo w\) then set \(\ell = |u|\), \(i = i + 1\) and
	    go back to Step~2.
\end{description}

One can easily observe that all Steps are effective and \(P\) ultimately
 halts for any input word \(w\) because the length of the word \(u\) in Step~4 is strictly
 increasing until \(u = w\) or \(w \slo u\).
Thus the language \(L \defeq \{ w \in A^* \mid P \text{ accepts } w \}\) is
 recursive.
 Moreover, \(L\) satisfies the following properties: (1) \(L \not\subseteq R\) for any regular co-infinite language
 because by Step~(4--5) \(P\) accepts some word
 \(w \notin R\), and (2) \(\d(L) = 0\); by Step~(5--6) and the length of \(u\) is strictly increasing,
 \(P\) rejects every word in \(A^n\) except
 for one single word \(u\), for each \(n\).
 Clearly, (2) implies \(\d(L) = 0\), and
 (1) implies \(\um(L) = 1\) since every language \(R\)
 with \(\cd(R) < 1\) is co-infinite.
 \end{proof} 
The second type of languages (Section~\ref{sec:amb}) is inherently ambiguous languages
and the third type of languages (Section~\ref{sec:kemp}) includes Kemp's language \(\kemp\) whose density is transcendental.
The last type of languages (Section~\ref{sec:gap}) is languages with full
\(\reg\)-gap, \ie, strongly
 \(\reg\)-immeasurable languages.

\subsection{Null Context-free Languages}\label{sec:null}
First we consider the following language with constraints on the number
of occurrences of letters, which is a very typical example of a non-regular but context-free language.
\begin{definition}\label{def:langeq}\upshape
For an alphabet \(A\) and letters \(a,b \in A\) such that \(a \neq b\),
 we define
\[
L_A(a,b) \defeq \{ w \in A^* \mid |w|_a = |w|_b\}.
\]
\end{definition}
\begin{theorem}\label{thm:ab}
\(L_A(a,b)\) is \(\reg\)-measurable where \(A = \{a,b\}\).
\end{theorem}
\iffull
\begin{proof}
\else
\begin{proof}[sketch]
\fi
It is enough to show that the complement \(L = \overline{L(a,b)}\)
 satisfies \(\lm(L) = 1\).
For each \(k \geq 1\), we define
\[
  L_k \defeq \{ w \in A^* \mid |w|_a \neq |w|_b \mod k\}.
\]
 Clearly, \(L_k \subseteq L\) holds.
 Each \(L_k\) is recognised by a \(k\)-states deterministic automaton
 \[
  \CA_k =
 (Q_k = \{q_0,\ldots, q_{k-1}\}, \Delta_k: Q_k \times A \rightarrow Q_k, q_0, Q_k \setminus \{q_0\})
 \]
 where
 \[
  \Delta_k(q_i, a) = q_{i+1 \!\!\!\!\mod k} \qquad \Delta_k(q_i, b) = q_{i-1
 \!\!\!\!\mod k} \quad (\text{ for each } i \in \{0, \ldots, k-1\}),
 \]
 \(q_0\) is the initial state, and any other state \(q \in Q_k \setminus
 \{q_0\}\) is a final 
 \iffull
 state
 (the case \(k = 3\) is depicted in Fig~\ref{fig:A3}).
 \fi
 \iffull
\begin{figure}[t]
\centering\includegraphics[width=.2\columnwidth]{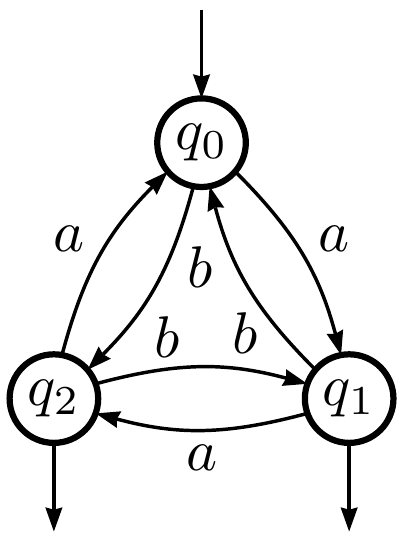}
\caption{The deterministic automaton \(\CA_3\) in the Proof of
 Theorem~\ref{thm:ab}.
 Here, the state \(q_0\) having unlabelled incoming arrow is
 initial and the states \(q_1, q_2\) having unlabelled
 outgoing arrow are final.}
\label{fig:A3}
\end{figure}
 The adjacency matrix of \(\CA_k\) is 
\begin{align*}
  M_k = \! \begin{bmatrix}
		 0&1&0&\cdots &\cdots &1\\
		 1&0&1&\ddots &&\vdots \\
		 0&1&\ddots &\ddots &\ddots &\vdots \\
		 \vdots &\ddots &\ddots &\ddots &1&0\\
		 \vdots &&\ddots &1&0&1\\
		 1&\cdots &\cdots &0&1&0
	   \end{bmatrix}
\! = E_k + E_k^{k-1} \text{ where }
E_k = \! \begin{bmatrix}
		 0&0&0&\cdots &\cdots &1\\
		 1&0&0&\ddots &&\vdots \\
		 0&1&\ddots &\ddots &\ddots &\vdots \\
		 \vdots &\ddots &\ddots &\ddots &0&0\\
		 \vdots &&\ddots &1&0&0\\
		 0 &\cdots &\cdots &0&1&0
       \end{bmatrix}\!.
\end{align*}

\(M_k\) is a special case of \emph{circulant matrices}.
 A \(k\)-dimensional circulant matrix \(C_k\) is a matrix that can be represented by a
 polynomial of \(E_k\):
 \[
  C_k = p(E_k) = \sum_{n = 0}^{k-1} c_n E_k^n
 \]
 and it is well
 known that \(C_k\) can be diagonalised as, for a \(k\)-th root of unity \(\xi_k = e^{-\frac{2
 \pi i}{k}}\) (where \(i\) is the imaginary unit),
 \[
 \frac{1}{\sqrt{k}} F_k^H \cdot C_k \cdot \frac{1}{\sqrt{k}} F_k = \mathrm{diag}(p(1),
 p(\xi_k^{-1}), p(\xi_k^{-2}), \ldots, p(\xi_k^{-(k-1)}))
 \]
 where \(F_k = (f_{n,m})\) with \(f_{n,m} = \xi_k^{(n-1)(m-1)}\)
 (for \(1 \leq n, m \leq k\)) is the \(k\)-dimensional \emph{Fourier matrix}, \(F_k^H\) is its Hermitian transpose 
 and \(\mathrm{diag}(\lambda_1, \cdots, \lambda_k)\) is the diagonal matrix
 whose \(n\)-th diagonal element is \(\lambda_n\) (for \(1 \leq n \leq
 k\)) (\cf Section~5.2.1~of~\cite{spectra}).
 Hence, in the case of \(M_k = p_{\CA_k}(E_k) = E_k + E_k^{k-1}\), we have
\begin{align}
 \frac{1}{\sqrt{k}} F_k^H \cdot M_k \cdot \frac{1}{\sqrt{k}} F_k = 
 \mathrm{diag}(2, \xi_k^{-1}+\xi_k, \xi_k^{-2}+\xi_k^{2}, \ldots,
 \xi_k^{-(k-1)}+\xi_k^{k-1}) \label{eq:diag}
\end{align}
 because, for any \(n \geq 0\), \(p_{\CA_k}(\xi_k^{-n}) = \xi_k^{-n} + \xi_k^{-n(k-1)} = \xi_k^{-n} +
 \xi_k^{n}\) holds.

 Let \(\Lambda_k = \mathrm{diag}(2, \xi_k^{-1}+\xi_k, \xi_k^{-2}+\xi_k^{2}, \ldots, \xi_k^{-(k-1)}+\xi_k^{k-1})\).
 Because \(\CA_k\) is deterministic
 and the final states are all but \(q_0\), the number of words of length \(n\)
 in \(L_k\) is exactly the number of paths from \(q_0\) to any other
 state in \(\CA_k\).
 For the \(k\)-dimensional vectors \(\bm{e} = (1,0,0,\ldots,0)\) and \(\bm{1} =
 (1, 1, 1, \ldots, 1)\), from Equation~\eqref{eq:diag} we have
\begin{align}
 & \card{L_k \cap A^n} = \bm{e} \cdot M_k^n \cdot (\bm{1} - \bm{e})^T \nonumber\\
 & = \frac{1}{k} \bm{e} \cdot F_k \cdot \Lambda_k^n \cdot F_k^H
 (\bm{1}-\bm{e})^T \nonumber\\
 &= \frac{1}{k} \bm{1} \cdot \Lambda_k^n \cdot \left(k-1, \sum_{j = 1}^{k-1} \xi_k^{-j},
 \sum_{j = 1}^{k-1} \xi_k^{-2j}, \ldots, \sum_{j = 1}^{-(k-1)}
 \xi_k^{-(k-1)j}\right)^T \nonumber\\
 &= \frac{1}{k} \left( 2^n (k-1) + (\xi_k^{-1} + \xi_k)^n \sum_{j = 1}^{k-1}
 \xi_k^{-j} + \cdots + (\xi_k^{-(k-1)} + \xi_k^{k-1})^n \sum_{j = 1}^{k-1}
 \xi_k^{-(k-1)j} \right). \label{eq:sum}
\end{align}
If \(k\) is odd \(k = 2m+1\), then for any \(1 \leq j \leq k-1\),
 \(\xi_k^{-j}+\xi_k^{j}\) is a real number whose absolute value is strictly smaller than
 \(2\); because \(\xi_k^{-j}\) is the complex conjugate of \(\xi_k^j\)
 and hence \(|\xi_k^{-j} + \xi_k^j| = |2\mathrm{Re}(\xi_k^j)| < 2\) for odd \(k\).
 Hence from Equation~\eqref{eq:sum} we can deduce that
\fi
\[
  \card{L_k \cap A^n} = \frac{k-1}{k} 2^n + o(2^n)
\]
 where \(o(2^n)\) means some function such that
 \(\lim_{n \rightarrow \infty} o(2^n)/2^n = 0\).
 Thus we have \(\d(L_k) = \frac{k-1}{k}\) for odd \(k = 2m + 1\), which tends to
 \(1\) if \(k\) tends to infinity, \ie, \(\m(L) = 1\).
 This completes the proof. \qedd
\end{proof}

By Theorem~\ref{thm:ab}, it is also true that any subset of
\(L_{\{a,b\}}(a,b)\) is \(\reg\)-measurable.
In particular, we have:
\begin{corollary}
The semi-Dyck language \(\dyck \subseteq L_{\{a,b\}}(a,b)\) is \(\reg\)-measurable.
\end{corollary}

The next example is the set of all palindromes.
\begin{theorem}\label{thm:pal}
\(
 \pal_A \defeq \{ w \in A^* \mid  w = \rev(w) \}
\)
 is \(\reg\)-measurable.
\end{theorem}
\iffull
\begin{proof}
Because the case \(\card{A} = 1\) is trivial (\(\pal_A = A^*\)), we assume that
 \(\card{A} \geq 2\).
It is enough to show that the complement \(\overline{\pal_A}\) is \(\reg\)-measurable.

For each \(k \geq 1\), we define
\[
  L_k \defeq \{ w_1 A^* w_2 \mid w_1, w_2 \in A^k, w_1 \neq \rev(w_2) \}.
\]
 One can easily observe that \(L_k \subseteq \overline{\pal_A}\) for each \(k \geq 1\).
 Moreover, for any \(n > 2k\), the number of words in \(L_k\) of length
 \(n\) is
\[
 \card{L_k \cap A^n} = \card{A}^k \cdot \card{A}^{n-2k} \cdot
 (\card{A}^k - 1) =  \card{A}^{n} - \card{A}^{n-k} . \label{eq:pal}
\]
 From this we can conclude that \(\d(L_k) = 1 - \card{A}^{-k}\)
and it tends to \(1\) if \(k\) tends to infinity.
Thus we have \(\m(\overline{\pal_A}) = 1\). 
\end{proof}
\fi
 
\subsection{Some Inherently Ambiguous Languages}\label{sec:amb}
There are \(\reg\)-measurable inherently ambiguous context-free
languages.
Since every \emph{bounded language} \(L \subseteq w_1^* \cdots
w_k^*\) is trivially \(\reg\)-measurable (\(\m(L) = 0\)),
a typical example of an inherently ambiguous context-free language 
\(\{ a^i b^j c^k \mid i = j \text{ or } i = k\}\)
is \(\reg\)-measurable.

Some more complex examples of inherently ambiguous languages are the following languages with
constraints on the number of occurrences of letters investigated by
Flajolet~\cite{flajolet}:
\begin{align*}
 \othree &\defeq \{ w \in \{a,b,c\}^* \mid |w|_a = |w|_b \text{ or } |w|_a
 = |w|_c \},\\
 \ofour &\defeq \{ w \in \{x,\bar{x},y, \bar{y}\}^* \mid |w|_{x} = |w|_{\bar{x}} \text{ or } |w|_y
 = |w|_{\bar{y}} \}.
\end{align*}

\begin{theorem}\label{thm:flaj}
\(\othree\) and \(\ofour\) are \(\reg\)-measurable.
\end{theorem}
\iffull
\begin{proof}
\else
\begin{proof}[sketch]
\fi
Let \(A = \{a,b,c\}\).
For the case \(\othree\), in a very similar way to Theorem~\ref{thm:ab}, we can construct  a
sequence of automata \((\CA_k^{ab})_{k \in \nat}\) such that
each automaton \(\CA_k^{ab}\) satisfies \(L(\CA_k^{ab}) \subseteq
\overline{L_A(a,b)}\)
\iffull
and its adjacency matrix is of the form
\begin{align*}
  M_k^{ab} = M_k + I_k = \begin{bmatrix}
		 1&1&0&\cdots &\cdots &1\\
		 1&1&1&\ddots &&\vdots \\
		 0&1&\ddots &\ddots &\ddots &\vdots \\
		 \vdots &\ddots &\ddots &\ddots &1&0\\
		 \vdots &&\ddots &1&1&1\\
		 1&\cdots &\cdots &0&1&1
		\end{bmatrix}
\end{align*}
where \(M_k\) is the adjacency matrix stated in Theorem~\ref{thm:ab} and
 \(I_k\) is the \(k\)-dimensional identity matrix.
\fi
The automaton \(\CA_k^{ab}\) is obtained by just adding self-loop
 labeled by \(c\) for each state \(q \in Q_k\) of \(\CA_k\) in Theorem~\ref{thm:ab}.
This sequence of automata ensures that the language \(L_A(a,b)\) is \(\reg\)-measurable (\(\um(L_A(a,b)) = 0\), in particular).
 The same argument is applicable to the language \(L_A(a,c)\), thus
 these union \(\othree = L_A(a,b) \cup L_A(a,c)\) is also
\iffull
\(\reg\)-measurable by Lemma ~\ref{lem:measure}.
\fi
The case \(\ofour\) can be achieved in the same manner. \qedd
 \end{proof}
 
Next we consider the so-called \emph{Goldstine language}
\[
 \gold \defeq \{ a^{n_1} b a^{n_2}b \cdots a^{n_p} b \mid p \geq 1, n_i \neq i
 \text{ for some } i \}.
\]
While \(\gold\) can be accepted by a non-deterministic pushdown
automaton, its generating function is not algebraic~\cite{FLAJOLET1987283} and thus it is
 an inherently ambiguous context-free language due to
 the well-known Chomsky--Sch\"utzenberger theorem stating that
 the generating function of every unambiguous context-free language is algebraic~\cite{CHOMSKY1963118}.

\begin{theorem}\label{thm:gol}
\(\gold\) is \(\reg\)-measurable.
\end{theorem}
\begin{proof}
Let \(A = \{a,b\}\).
Observe that
\(\gold \subseteq A^* b \) and \(\um(\gold) \leq \d(A^* b) = 1/2\).
Let
\[
 L_\gold = \{ u \in A^* \mid u A^*\{b\} \cap \overline{\gold} = \emptyset \}
\]
 be the set of all forbidden prefixes of the complement \(\overline{\gold}\).
For each \(k \geq 1\), we define
\[
 L_k \defeq \{ u A^*\{b\} \mid u \in L_\gold \cap A^k \}.
\]
 If a word \(u\) is in \(L_\gold\), then by definition of \(L_\gold\),
 \(uvb\) is always in \(\gold\) for any word \(v\), thus \(L_k \subseteq
 G\) holds for each \(k\).
 Any word in \(\overline{L_\gold} = A^* \setminus L_\gold\)
 is a prefix of the infinite word
 \(
 a^{n_1} b a^{n_2} b a^{n_3} b \cdots \; (n_i = i \text{ for each
 } i \in \nat)
 \)
 thus \(\card{L_\gold \cap A^n} = \card{A^n} - 1\) holds for each \(n \geq 1\).
 Hence we have
\begin{align*}
\d(L_k) &= \lim_{n \rightarrow \infty} \! \frac{\card{L_k \cap A^n}}{\card{A^n}} =
 \lim_{n \rightarrow \infty} \! \frac{(\card{A^k}-1) \cdot \card{A^{n-k-1}}}{\card{A^n}}
 \\
 &=
 (\card{A}^k-1) \cdot \card{A}^{-k-1}
 = 2^{-1}-2^{-k-1}.
\end{align*}
This implies that \(\d(L_k)\) tends to \(1/2\).
 Thus \(\m(\gold) = 1/2\). \qedd
\end{proof}

In general, for an infinite word \(w \in A^\omega\), the set
\[
 \mathrm{Copref}(w) \defeq A^* \setminus \{ u \in A^* \mid u \text{ is a prefix of } w \}
\]
is called the \emph{coprefix language of \(w\)}.
The proof of Theorem~\ref{thm:gol} uses a key property that \(\gold\)
can be characterised by using
the coprefix language of the infinite word \(w = a^{n_1} b a^{n_2} b
a^{n_3} b \cdots\) as \(\gold = \mathrm{Copref}(w) \cap
\{a,b\}^*\{b\}\) which was pointed out in~\cite{copref}. Thus by the same argument, we can say that any
coprefix language \(L\) is \(\reg\)-measurable (\(\m(L) = 1\), in
particular).

For coprefix languages, the following nice ``gap theorem'' holds.
\begin{theorem}[Autebert--Flajolet--Gabarro~\cite{copref}]
Let \(w \in A^\omega\) be an infinite word generated by an iterated morphism, \ie,
 \(w = h(w) = h^\omega(a)\) for some monoid morphism \(h: A^* \rightarrow
 A^*\) and letter \(a \in A\).
Then for the coprefix language \(L = \mathrm{Copref}(w)\) there are only
 two possibilities:
\begin{enumerate}
 \item \(L\) is a regular language. 
 \item \(L\) is an inherently ambiguous context-free language.
\end{enumerate}
\end{theorem}
This means that we can construct, by finding some suitable morphism
\(h\), many examples of inherently ambiguous context-free languages.

\subsection{\(\kemp\): A Language with Transcendental Density}\label{sec:kemp}
We now show the fact that the language \(\kemp\) defined by
Kemp~\cite{Kemp1980} (recall that the definition of \(\kemp\) appeared
in Therem~\ref{thm:kemp}) is \(\reg\)-measurable.
We will actually show a more general result regarding the following
type of languages.
\begin{definition}\upshape
Let \(L \subseteq A^*\) be a language and \(c \notin A\) be a letter.
We call the language \(L \{c\} (A \cup \{c\})^*\) over \(A \cup \{c\}\)
\emph{suffix extension of \(L\) by \(c\)}.
\end{definition}
\begin{theorem}\label{thm:suff}
The suffix extension \(L' \subseteq (A \cup \{c\})^*\) of any language
 \(L \subseteq A^*\) by \(c \notin A\) is \(\reg\)-measurable.
\end{theorem}
 \begin{proof}
Let \(B = A \cup \{c\}\) and \(k = \card{B}\).
We first show that \(L'\) has a natural density.
For any words \(u, v \in L\) with \(u \neq v\),
 two languages \(u\{c\} B^*\) and
 \(v\{c\} B^*\) are disjoint, and clearly
 \[
 \card{u\{c\} B^* \cap B^n}/\card{B^n}
  = \card{u\{c\} B^{n-|u|-1}}/\card{B^n}
  = k^{n-|u|-1}/k^{n}
  = k^{-(|u|+1)}
 \]
 holds for \(n > |u|\)
 thus \(\pd{B}(u\{c\}B^*) = k^{-(|u|+1)}\).
The natural density of \(L'\) is
\begin{align}
 \pd{B}(L') & = \lim_{n \rightarrow \infty} \frac{\card{L' \cap
 B^n}}{\card{B^n}}
 = \lim_{n \rightarrow \infty} \frac{\card{ \bigcup_{w \in L} (w \{c\} B^*
 \cap B^n)}}{\card{B^n}} \nonumber\\
 &= \lim_{n \rightarrow \infty} \frac{\sum_{w \in L} \card{w \{c\} B^*
 \cap B^n}}{\card{B^n}} = \lim_{n \rightarrow \infty} \sum_{w \in
 (L \cap A^{<n} )} k^{-(|w|+1)}. \label{eq:lim}
\end{align}
 Because the sequence \((\sum_{w \in
 (L \cap A^{<n} )} k^{-(|w|+1)})_{n \in \nat}\) is non-decreasing and bounded above by
 \(1\), the limit \eqref{eq:lim} exists, say \(\pd{B}(L') = \alpha\).

For each \(n \in \nat\), the language
 \(L_n \defeq \bigcup_{w \in L \cap A^{<n}} w\{c\}B^*\) is regular (since \(L \cap
 A^{<n}\) is finite), \(L_n \subseteq L'\) and
 \(\pd{B}(L_n) = \sum_{w \in (L \cap A^{<n} )} k^{-(|w|+1)}\).
 Hence \(\lm(L') = \alpha\).
 By similar argument,
 for each \(n \in \nat\), we can claim that the language
 \(K_n \defeq B^* \setminus \bigcup_{w \in \overline{L} \cap A^{<n}}
 w\{c\}B^* \)
 satisfies \(K_n \supseteq L'\) and \(\pd{B}(K_n)\) tends to \(\alpha\)
 if \(n\) tends to infinity.
 Thus \(\m(L') = \alpha\).\qedd
 \end{proof}
Since \(\kemp\) is the suffix extensions of the union \(S_1 \cup S_2\) in
Theorem~\ref{thm:kemp}, we have:
\begin{corollary}
\(\kemp\) is \(\reg\)-measurable.
\end{corollary}

\begin{remark}\label{rem:relax}
Theorem~\ref{thm:suff} indicates that \(\reg\)-measurability is a quite
relaxed property in some sense:
even for a non-recursively-enumerable language,
its suffix extension is still non-recursively-enumerable but \(\reg\)-measurable.
  Moreover, because the class of recursively enumerable languages is
 just a countable set, there exist \emph{uncountably many}
 \(\reg\)-measurable non-recursively-enumerable languages.
\end{remark}

The same proof method works for the \emph{prefix extension} and the
\emph{infix extension} (see the full version~\cite{full} for details).

\iffull
The same proof method works for the \emph{prefix extension} and the \emph{infix extension}.
\begin{theorem}\label{thm:ext}
Let \(c \notin A\) and \(A' = A \cup \{c\}\).
The prefix extension \(L' = A'^* \{c\} L\) of any language
 \(L \subseteq A^*\) is \(\reg\)-measurable.
Also, the infix extension \(L'' = A'^* \{c\} L \{c\} A'^*\) of any language
 \(L \subseteq A^*\) is \(\reg\)-measurable,
 \(\m(L'') = 0\) if \(L = \emptyset\),
 \(\m(L'') = 1\) otherwise, in particular.
\end{theorem}
\begin{proof}
 The prefix extension of \(L\) is just the reverse
 of the suffix extension of \(L\), the same proof method trivially
 works.
 For the infix extension
 \(L'' = A'^* \{c\} L \{c\} A'^*\),
 if \(L = \emptyset\) then \(L''\) is also empty and thus \(\m(L'') =
 0\).
 Further, if \(L \neq \emptyset\) then
 there is a word \(w \in L\) and thus
 \(A'^* c w c A'^* \subseteq  L''\) holds,
 which means that \(\pd{A'}(A'^* c w c A'^*) = 1\) by the infinite monkey
 theorem and we have \(\m(L'') = 1\).
\end{proof}
\fi

\subsection{Languages with Full \(\reg\)-Gap}\label{sec:gap}
In Section~\ref{sec:null}, we showed that the language \(L_{\{a,b\}}(a,b)\) is
\(\reg\)-measurable. On the other hand,
by the result of Eisman--Ravikumar~\cite{approx}, we will know that the closely related language
\[
 \maj \defeq \{ w \in \{a,b\}^* \mid |w|_a > |w|_b \},
\]
sometimes called the \emph{majority language}, is not \(\reg\)-measurable.
This contrast is interesting.

\begin{theorem}[Eisman--Ravikumar~\cite{approx,Eisman2011}]
Let \(A = \{a,b\}\) and \(L \subseteq A^*\) be a regular language.
Then \(\maj \subseteq L\) implies
\iffull
 \[
 \limsup_{n \rightarrow \infty}
 \{\card{\overline{L} \cap A^n}/\card{A^n}\} = 0.
 \]
 \fi
\end{theorem}
One can easily observe that \(\limsup_{n \rightarrow \infty}
\{\card{\overline{L} \cap A^n}/\card{A^n}\} = 0\)
if and only if \(\d(\overline{L}) = 0\), which means that any regular
superset of \(\maj\) is co-null.
Thus the above theorem implies
that both \(\maj\) and \(\overline{\maj}\) are \(\reg^+\)-immune, hence
we have:
\begin{corollary}
\(\maj\) has full \(\reg\)-gap.
\end{corollary}

By using the infinite monkey theorem and some probabilistic arguments,
we can generalise the previous theorem as follows.
\begin{theorem}\label{thm:pmaj}
For any \(m \geq 1\), the following language over \(A = \{a,b\}\)
\[
  \maj_m \defeq \{ w \in A^* \mid |w|_a > m \cdot |w|_b \}
\]
 has full \(\reg\)-gap, and
\(\d(\maj_m) = 1/2\) if \(m = 1\) otherwise \(\d(\maj_m) = 0\).
\end{theorem}
\iffull
\begin{proof}
\else
\begin{proof}[sketch]
\fi
First we prove that any non-null regular language \(L\) can not be a
 subset of \(M_m\).
Let \(\eta: A^* \rightarrow M\) be the syntactic morphism \(\eta\) and
 monoid \(M\) of \(L\), and let
\(
 c = \max_{m \in M} \min_{w \in \eta^{-1}(m)} |w|
\)
 (this is well-defined natural number since \(M\) is finite).
 By the infinite monkey theorem, \(L\) is not null implies that \(L\)
 has no forbidden word, and thus for the word
 \(b^{2c}\) there exist two words \(x\) and \(y\) such that \(x b^{2c} y\) is
 in \(L\).
 We can assume that \(|x|,|y| \leq c\) without loss of generality by the
 definition of \(c\), which implies
 \(|x b^{2c} y|_a \leq |x|+|y| = 2c \leq |x b^{2c} y|_b\) hence
 \(x b^{2c} y \notin \maj_m\).
 Thus \(L \not\subseteq \maj_m\) and \(\lm(\maj_m) = 0\).
 By using same argument, we can prove that
 \(\um(\maj_m) = 1\) and hence \(\maj_m\) has full \(\reg\)-gap.

In the case \(m = 1\), \(\d(\maj_1) = \d(\maj) = 1/2\) is obvious.
\iffull 
It is enough to show that \(\d(\maj_2) = 0\) holds (since \(\maj_m
 \subseteq \maj_2\) for any \(m \leq 2\)).
 Indeed, we have
\begin{align*}
 \d(\maj_2) &= \lim_{n \rightarrow \infty}
 \frac{\card{\{ w \in A^n \mid |w|_a > 2|w|_b \}}}{2^n}
 = \lim_{n \rightarrow \infty} \frac{\card{\{ w \in A^n \mid |w|_a > 2n/3 \}}}{2^n}
 \\
&= \lim_{n \rightarrow \infty} \pr(|\overline{X}_n - n/2| > n/6) = 0
\end{align*}
 where
 \(\pr(|\overline{X}_n - n/2| > n/6)\) means the probability that the absolute
 value of the difference of
 the number \(\overline{X}_n\) of the occurrences of \(a\)'s in a
 randomly chosen word of length \(n\) and its mean value \(n/2\) is
 larger than \(n/6\); its tends to zero by the weak law of large
 numbers.
\fi
\end{proof}

 \section{\(\reg\)-Immesurability of Primitive Words}\label{sec:prim}
A non-empty word \(w \in A^+\) is said to be primitive if
\(u^n = w\) implies \(u = w\) for any
\(u \in A^+\) and \(n \in \nat\).
The set of all primitive words over \(A\) is denoted by \(\prim_A\).
Because the case \(\card{A} = 1\) is meaningless (\(\prim_A = A\) in
this case), hereafter we
always assume \(\card{A} \geq 2\).
Whether \(\prim_A\) is context-free or not is a well-known long-standing
open problem posed by D{\"{o}}m{\"{o}}si, Horv{\'{a}}th and
Ito~\cite{Domosi1991}.
Reis and Shyr~\cite{Reis} proved \(\prim_A^2 = A^+ \setminus \{ a^n \mid a
\in A, n \neq 2 \}\), which intuitively means that every
non-empty word \(w\) not a power of a letter is a product of two
primitive words.
From this result one may think that \(\prim_A\) is ``very large'' in
some sense.
Actually, \(\prim_A\) is
somewhat ``large'' (it is dense in the sense of
Definition~\ref{def:null}), but we can show more stronger property as
\iffull
follows.
\fi
\begin{theorem}\label{thm:prim1}
\(\d(\prim_A) = 1\).
\end{theorem}
\iffull
\begin{proof}
It is enough to show that \(\d(\overline{\prim_A}) = 0\) holds.
One can easily observe that any natural number \(n \in \nat\)
 has at most \(2\sqrt{n}\) divisors.
In addition, for any non-primitive word \(w = v^m\) of length \(n\)
 is uniquely determined by \(v\) (since \(m = n/|v|\)) and \(|v| \leq n/2\).
Hence the number of non-primitive words of length \(n\) satisfies
\[
 \card{\overline{\prim_A} \cap A^n} \leq 2 \sqrt{n} \sum_{i =
 0}^{\lfloor n/2 \rfloor } \card{A^i} \leq 2 \sqrt{n} \cdot
 \card{A}^{\lfloor n/2 \rfloor + 1}.
\]
 By using the above estimation, we can deduce that
\[
 \frac{\card{\overline{\prim_A} \cap A^n}}{\card{A^n}}
 \leq \frac{2 \sqrt{n} \cdot \card{A}^{\lfloor n/2 \rfloor +
 1}}{\card{A}^n} \leq \frac{2 \sqrt{n}}{\card{A}^{n/2 - 1}}
\]
 and it tends to \(0\) if \(n\) tends to infinity (since we assume
 \(\card{A} \geq 2\)). Thus \(\d(\overline{\prim_A}) = 0\).
\end{proof}
\fi

While \(\prim_A\) is ``very large'' (co-null) as stated above, we can also prove that
 \(\prim_A\) is \(\reg^+\)-immune.
The proof relies on an analysis of the structure of the syntactic monoid
of a non-null regular language.
We assume that the reader has a basic
 knowledge of semigroup theory (\cf \cite{pin}):
 Green's
relations \(\CJ, \CR, \CL, \CH\) and a direct consequence of Green's
theorem (an \(\CH\)-class \(H\) in a semigroup \(S\) is a subgroup of \(S\) if
and only if \(H\) contains an idempotent), in particular.
\begin{theorem}\label{thm:prim2}
Any non-null regular language contains infinitely many non-primitive
 words, and hence \(\lm(\prim_A) = 0\).
\end{theorem}
\begin{proof}
Let \(L\) be a regular language over \(A\) with a positive density
 \(\d(L) > 0\).
We consider \(\eta: A^* \rightarrow M\) the syntactic morphism \(\eta\) and
 the syntactic monoid \(M\) of \(L\), and let \(S\) be a subset of \(M\)
 satisfying \(\eta^{-1}(S) = L\).
 \(L\) is regular means that \(M\) is finite, and hence \(M\) has at
 least one \(\leq_\CJ\)-minimal element.

We first show that \(S\) contains a \(\leq_\CJ\)-minimal element \(t\).
This is rather clear because, for any non-\(\leq_\CJ\)-minimal element
 \(s\), its language \(\eta^{-1}(s) \subseteq A^*\) is null: \(s\) is
 non-\(\leq_\CJ\)-minimal means that there
 is an other element \(t\) such that \(t <_\CJ s\)
 (\ie, \(M t M \subsetneq M s M\)), whence \(s \notin M t M\)
 which implies that any word \(w \in \eta^{-1}(t)\) is a forbidden word
 of \(\eta^{-1}(s)\). Thus by the infinite monkey theorem
 \(\eta^{-1}(s)\) is null.

Clearly, we have \(t^n \leq_\CJ t\) and thus
 \(t \, \CJ \, t^n\) holds for any \(n > 1\) by the
 \(\leq_\CJ\)-minimality of \(t\).
 \(t \, \CJ \, t^n\) implies that there is a pair of words \(x,y\) such that \(x t^n y =
 t\). Since \(M\) is finite, 
 \(x^m\) is an idempotent for some \(m > 0\)
 (\ie, \(x^{2m} = x^m\)).
 Thus we obtain \(t = x t^n y = x (t) t^{n-1} y = x^2 (t) (t^{n-1} y)^2 = \cdots = x^m t (t^{n-1} y)^m =
 x^m x^m t (t^{n-1} y)^m = x^m t\) whence
 \(t = t^n (y (t^{n-1} y)^{m-1})\).
 It follows that \(t \, \CR \,
 t^n\).
 Dually, we also obtain \(t \, \CL \, t^n\) and hence we can deduce that
 \(t \, \CH \, t^n\) holds.
By the finiteness of \(M\), there exists some \(n > 0\) such that
 \(t^n\) is an idempotent.
Thanks to Green's theorem, the \(\CH\)-equivalent class \(H_t\) of \(t\)
 is a subgroup of \(M\) with the identity element \(t^n\).
 Because
 \(\eta\) is surjective, we can take a word \(w'\) from
 \(\eta^{-1}(t)\).
 Let \(t' = \eta(w' a) = t \eta(a)\) for some letter \(a \in A\), then by the
 \(\leq_\CJ\)-minimality of \(t\),
 we can take some words \(x,y \in A^*\)
 so that \(\eta(x w' a y) = \eta(x) t' \eta(y) = t\). Hence we can
 deduce that \(\eta^{-1}(t)\) contains a non-empty word \(w = x w' a
 y\).
 Then for any \(\varepsilon \neq w \in \eta^{-1}(t)\) and \(m \geq 1\), we have
\[
 \eta(w^{mn+1}) = t^{mn+1} = (t^{n})^m \cdot t
  = t \in S
\]
 which means that \(L \supseteq \eta^{-1}(t)\)
 contains infinitely many non-primitive words  \(w^{mn+1}\). \qedd
 \end{proof}

\begin{corollary}[of Theorem~\ref{thm:prim1}~and~\ref{thm:prim2}]
\(\prim_A\) has full \(\reg\)-gap.
\end{corollary}

\begin{remark}\label{rem:prim2}
We emphasise that the assumption ``\(L\) is non-null'' in Theorem~\ref{thm:prim2} is quite tight,
 since a slightly weaker assumption ``\(L\) is of exponential growth''
 (\ie, \(\card{L \cap A^n}\) is exponential for \(n\)) does not imply
 that \(L\) contains non-primitive words.
 A trivial counterexample is
 \(L_0 = \{a,b\}^* c\) over \(A = \{a,b,c\}\):
 \(\card{L_0 \cap A^n} = 2^{n-1} \, (n \geq 1)\) is exponential but \(L_0\)
 only consists of primitive words.
 \(L_0\) has a \(cc\) as a forbidden word, hence it is null by the
 infinite monkey theorem. Thus \(L_0\) is not a counterexample of Theorem~\ref{thm:prim2}.
\end{remark}
 \section{Conclusion and Open Problems}\label{sec:fut}
In this paper we proposed \(\reg\)-measurability and showed that several
context-free languages are
\(\reg\)-measurable, excluding \(\maj_m\).
Interestingly, it is shown that, like \(\gold\) and \(\kemp\),
languages that have been considered as complex from a combinatorial
viewpoint are, actually, easy to asymptotically approximate by
regular languages.
It is also interesting that a modified majority language \(\maj_2\) is
just a deterministic context-free but it is complex from a
measure theoretic viewpoint.
Its complement \(\overline{\maj_2}\) is also deterministic context-free, and actually it
is co-null but \(\reg^+\)-immune (\ie, has full
\(\reg\)-gap).
This means that \(\overline{\maj_2}\) is as complex as \(\prim_A\) from
a viewpoint of \(\reg\)-measurability.

The following fundamental problems are still open and we consider these
to be future work.
\begin{problem}\label{prob:cflmeasurable}
Can we give an alternative characterisation of the null (resp. co-null) context-free
 languages (like Theorem~\ref{thm:null})? 
\end{problem}
\begin{problem}\label{prob:cflmeasurable}
Can we give an alternative characterisation of the \(\reg\)-measurable context-free languages?
\end{problem}
\begin{problem}\label{prob:separate}
Can we find a language class that can ``separate'' \(\prim_A\)
 and \(\cfl\)?
\ie, is there \(\CC\) such that
 \(\prim_A\) has full \(\CC\)-gap but no co-null context-free language has
 full \(\CC\)-gap, or
 \(\prim_A\) is \(\CC\)-immeasurable but any co-null context-free language is
 \(\CC\)-measurable?
\end{problem}
The our results (Theorem~\ref{thm:pmaj},
\ref{thm:prim1}~and~\ref{thm:prim2})
tell us that the class \(\reg\) of regular languages can not separate
\(\prim_A\) and  \(\cfl\). However, it is still open whether the situation
is the same or not when \(\CC = \dcfl, \ucfl, \cfl\) or other extension of regular languages.
Notice that \emph{if} the answer of Problem~\ref{prob:separate} is
``yes'', then \(\prim_A\) is not context-free.\\

\noindent {\bf Acknowledgement:}
The author would like to thank Takanori Maehara (RIKEN AIP) and Fazekas
Szil\'ard (Akita University) whose helpful discussion
were an enormous help to me.
The author also thank to 
anonymous reviewers for many valuable comments.
This work was supported by JSPS KAKENHI Grant Number JP19K14582.

\end{document}